\documentclass[11pt]{article}

\usepackage[colorlinks,linkcolor=blue,citecolor=blue,urlcolor=blue]{hyperref}
\usepackage{amsmath,amsthm}

\usepackage{fullpage}
\usepackage{url} 
\newcommand{\cites}{\cite}

\newtheorem{theorem}{Theorem}

\newtheorem{lemma}[theorem]{Lemma}

\newcommand{\cN}{{\mathcal{N}}}
\newcommand{\cP}{{\mathcal{P}}}

\newcommand{\cR}{{\mathcal{R}}}

\newcommand{\tr}{\mathop{\mathrm{tr}}}
\newcommand{\rank}{\mathop{\mathrm{rank}}}

\def\l{\left}
\def\r{\right}
\renewcommand{\>}{\rangle}
\newcommand{\<}{\langle}
\def\ra{\rightarrow}
\def\la{\leftarrow}
\newcommand{\ket}[1]{|#1\rangle}

\newcommand{\proj}[1]{\left|#1\right\>\!\left\<#1\right|}

\newcommand{\ot}{\otimes}
\newcommand{\eps}{\epsilon}

\def\geqclean{\stackrel{\!{\scriptstyle clean}}{\geq}\!}

\newcommand{\be}{\begin{equation}}
\newcommand{\ee}{\end{equation}}
\def\ba#1\ea{\begin{align}#1\end{align}}
\def\bit#1\eit{\begin{itemize}#1\end{itemize}}

\newcommand{\eq}[1]{(\ref{eq:#1})}
\newcommand{\thmref}[1]{Theorem~\ref{thm:#1}}

\newcommand{\lemref}[1]{Lemma~\ref{lem:#1}}
\newcommand{\secref}[1]{Section~\ref{sec:#1}}

\begin{document}

\title{Entanglement spread and clean resource inequalities}

\author{Aram W. Harrow\\
Department of Mathematics, 
University of Bristol,
Bristol, BS8 1TW, U.K.\\
harrow@gmail.com}

\maketitle

\begin{abstract}
  This article will examine states that superpose different amounts of
  entanglement and protocols that run in superposition but generate or
  consume different amounts of entanglement.  In both cases we find a
  uniquely quantum difficulty: entanglement cannot be conditionally
  discarded without either using communication or causing decoherence.

  I will first describe the problem of entanglement spread in states
  and operations, as well as some methods of dealing with it.  Then
  I'll describe three applications to problems that at first glance
  appear to be quite different: first, a reinterpretation of the old
  observation that creating $n$ partially entangled states from
  singlets requires $\theta(\sqrt{n})$ communication, but cannot
  itself be used to communicate; second, a new lower bound technique
  for communication complexity; third, an explanation of how to extend
  the quantum reverse Shannon theorem from tensor power sources to
  general sources.
\end{abstract}

\section{Introduction}
This paper will challenge the idea that, when it comes to
entanglement, more is always better.  While some resources in quantum
information theory, like use of communication channels, can be safely
discarded, entanglement cannot be kept in one branch of a
superposition and discarded in another without causing decoherence.
First, I will outline two specific challenges that this fact poses to
the traditional resource model of quantum information.

\subsection{Motivation: Coherent conditional execution of quantum
  communication protocols.}  

A productive way of understanding quantum communication protocols can
be to view quantum states and operations as {\em resources} and
protocols that convert one resource into another as {\em resource
  inequalities} (RIs)\cite{Bennett04,DHW05}.  So if $[q\ra q]$ represents a noiseless qubit
channel from Alice to Bob, $[qq]$ is an EPR pair and $[c\ra c]$ is a
noiseless classical bit channel (cbit) from Alice to Bob, then
teleportation can be expressed as $2[c\ra c]+[qq]\geq [q\ra q]$.  This
picture was formalized in \cite{DHW05}, which also proved many basic
intuitive facts about RIs.  For example, if $\alpha\geq\beta$ and
$\beta\geq\gamma$ for some resources $\alpha,\beta,\gamma$ then
$\alpha\geq \gamma$ as well.

However, there is another basic way of combining protocols
which works trivially in the classical case but fails in the quantum
case.  Suppose Alice and Bob each know a common bit $b$ and want to
perform protocol $\cN_b$ conditioned on the value of $b$.  This
happens often when communication protocols are embedded in larger
applications where the amount and type of communication is itself
input-dependent.  Let us call this conditional operation $b?\cN_1:\cN_0$,
following C notation.  Classically it is trivial to show that if
$\alpha\geq \<\cN_b\>$ for $b=0,1$ then $\alpha \geq
\<b?\cN_1:\cN_0\>$: Alice and Bob use $\alpha$ to perform either
$\cN_0$ or $\cN_1$ depending on their shared value of $b$.  However,
in a quantum protocol this might leak information about $b$ to the
environment, because the environment may be able to distinguish
$\cN_0$ and $\cN_1$ on some inputs.  In \secref{clean-RI} I will 
describe one possible solution to this problem, which will be called 
{\em clean resource inequalities}.  The main idea will be to discard
only (up to an asymptotically vanishing error) qubits in standard
states, such as $\ket{0}$, so that the environment cannot learn which
protocol is being run.

\subsection{Motivation: Non-asymptotic analysis of entanglement}

A second problem with the traditional resource framework arises in
quantifying pure state entanglement. The entanglement of a pure state
$\ket{\psi}^{AB}$ is usually said to be characterized by its {\em
  entropy of entanglement}, $E(\psi)$, which is defined as
$E(\psi) := S(\psi^A) = S(\psi^B)$.  Here $\psi:=\proj{\psi}$ is
the density matrix corresponding to $\psi$, $\psi^A := \tr_B \psi$ is
Alice's reduced density matrix (and similarly for $\psi^B$),
$S(\rho):=-\tr\rho\log\rho$ is the von Neumann entropy, and the base
of logs and exponentials will always be 2.  Asymptotically, the entropy
of entanglement characterizes the entanglement present in a state in
the following sense: given $\ket{\psi}^{\ot n}$ with
$E:=E(\ket{\psi})$, {\em entanglement concentration}\cite{BBPS96} can
produce $nE-o(n)$ maximally entangled states (i.e. $\ket{\Phi}^{\ot
  nE-o(n)}$, where
$\ket{\Phi}^{AB}:=\frac{1}{\sqrt{2}}(\ket{00}+\ket{11})$) with $o(1)$
error, while {\em entanglement dilution}\cite{BBPS96,LP99} can map
$\ket{\Phi}^{\ot nE+o(n)}$ to $\ket{\psi}^{\ot n}$ with $o(1)$ error
and $o(n)$ bits of classical communication.  Thus, to leading order in
the number of copies, tensor powers of entangled pure states can be
described by a single parameter: the entropy of entanglement.

In general, however, entangled pure states cannot be
fully described by the entropy of entanglement, even given free local
operations.  
For example, instead of von Neumann entropy, we can use 
R\`{e}nyi entropies.  For any $\alpha> 0$, define $E_\alpha(\psi) :=
S_\alpha(\psi^A) = \frac{1}{1-\alpha}\log\tr(\psi^A)^{\alpha}$, and
extend by continuity to $\alpha=0,1,\infty$.  In particular,
$E_0(\psi)$ is $\log\rank\psi^A$ and $E_\infty(\psi)=-\log
\|\psi^A\|_\infty$, where $\rank\psi^A$ is the number of non-zero
eigenvalues of $\psi_A$ and $\|\psi^A\|$ is its largest eigenvalue.
($S_0$ and $S_\infty$ are also called the max-entropy and min-entropy,
respectively.) 
Like $E(\psi)$, the $E_\alpha(\psi)$ are also
invariant under local unitaries, and non-increasing on average under
local operations and classical communicaton (LOCC).  
Thus they define
restrictions on entanglement transforms that may be more stringent
than those obtained from the non-increase of entropy of entanglement.
When communication is restricted then R\`{e}nyi entropies can further
limit possible entanglement transformations.

Following
\cite{HW02}, define the {\em entanglement spread} of a state
$\ket{\psi}^{AB}$ to be
\be\Delta(\psi) := \log \rank\psi^A + \log \|\psi^A\|_\infty.
\label{eq:spread-def}\ee
(Also in \cite{HW02} was the more general
$\Delta^{\alpha,\beta}(\psi) = H_\alpha(\psi^A) - H_\beta(\psi^A)$,
where $H_p(\rho) := \frac{1}{1-p}\log\tr\rho^p$ and $\alpha<\beta$.
We recover \eq{spread-def} by setting $\alpha=0$ and $\beta=\infty$.)
The entanglement spread is never negative and $\Delta(\psi)=0$ if and
only if $\psi$ has all non-zero Schmidt
coefficients equal, meaning it is a product state or a maximally
entangled state.  If $\Delta(\psi)>0$ then we say that
$\ket{\psi}$ is {\em partially entangled} (since it is neither
unentangled nor maximally entangled).  Since $\Delta(\psi_1\ot
\psi_2) = \Delta(\psi_1) + \Delta(\psi_2)$, it follows that if $\Delta(\psi)>0$,
then $\Delta(\psi^{\ot n})=\Theta(n)$.   However, this value is not
very robust: for any constant $\eps>0$, we can perturb
$\ket{\psi}^{\ot n}$ by $\eps$ and reduce
its spread to $\Theta(\sqrt{n})$.
To capture this insight we will also use the $\eps$-perturbed
entanglement spread 
$\Delta_\eps(\psi)$, which is defined for any $\eps\geq 0$ to be
(following \cite{HW02})
$$
\min\left\{ 
\log \tr P + \log \| P\psi^A P\|_\infty
: \tr P\psi^A \geq 1-\eps \right\},$$
where the minimization is over projectors $P$.  Note that
$\Delta_0(\psi)=\Delta(\psi)$.  Now
$\Delta_\eps({\psi}^{\ot n}) = \Theta(\sqrt{n})$ for any $\eps>0$, but
of course there are 
still states on $2n$ qubits, such as the even superposition between
$\ket{\Phi}^{\ot 
  n}$ and a product state, which have $\eps$-perturbed spread nearly
equal to $n$.

The main application of entanglement spread is using the following result
from \cite{HW02} to produce lower bounds on communication:
\begin{theorem}[Corollary 10 of \cite{HW02}]\label{thm:spread}
If $\ket{\phi}^{AB}$ is transformed using local operations and $C$
bits of classical communication (in either direction) into a state
that has fidelity $1-\eps$ with $\ket{\psi}$ then
$$C \geq \Delta_\delta(\psi) - \Delta_0(\phi) + 2\log(1-\delta),$$
where $\delta = (4\eps)^{1/8}$.
\end{theorem}
In particular, if we begin with maximally entangled states, then
preparing $\ket{\psi}$ to within a fidelity of $1-\eps$ requires
$\Delta_\delta(\psi) + 2\log(1-\delta)$ bits of communication.
This restriction holds even with an unlimited supply of EPR pairs, and
so problematizes the idea that maximally entangled states are a
good canonical form for the resource of pure state entanglement.  In
\secref{ent-cap} we will discuss alternative ways to quantify
entanglement as a resource.

\section{Dealing with entanglement spread}
In this section I propose two solutions to the above problems.  To
build quantum protocols that can be run in superposition, I will demand
that they discard only (approximately) standard states to the
environment, as I will describe in \secref{clean-RI}.  Then in
\secref{ent-cap}, I will propose measuring not just the maximum
amount of entanglement that can be created by a protocol, but instead
finding the range of entanglement that it can cleanly
generate/consume.  These approaches are not rigid rules, but rather
illustrate principles that can be adapted to diverse situations, as we
will see in \secref{applications}.
 
\subsection{Approach: clean resource inequalities}
\label{sec:clean-RI}

In \cite{DHW05}, a resource inequality $\alpha\geq\beta$ meant that
resource $\alpha$ could be approximately transformed using local
operations into $\beta$.
Allowing free local operations is standard practice in quantum
information theory, but discarding to the environment can be dangerous
when running different protocols in superposition.  Following and
extending \cite{HS05}, I will say that a {\em clean resource
  inequality} exists, and is denoted $\alpha \geqclean \beta$, when
$\alpha$ can be mapped to $\beta$ using only (up to error $\eps_n$
that goes to 0 as $n\ra\infty$) \bit
\item Local unitaries.
\item Adding ancillas initialized in the $\ket{0}$ states.
\item Discarding ancillas in the $\ket{0}$ state.
\item Discarding messages that have been sent through classical channels.
\item A dynamic resource $(\cN:\omega)$ (meaning, as defined in
  \cite{DHW05}, an operation $\cN$ constrained to act on average input
  $\omega$) may be used as a consumed resource only in a way that
  sends a constant state to the environment.  More formally, suppose
  we are given $(\cN^{A_1A_2A_3\ra BE}: \omega^{A_1}$.  Then we need to
  replace this resource with $(\cN^{A_1A_2A_3\ra B}:
  \tilde{\omega}^{A_1A_2})$ such that
  $\tilde{\omega}^{A_1}=\omega^{A_1}$ and $\cN(\sigma)^E$ is the same
  for all $\sigma$ satisfying
  $\sigma^{A_1A_2}=\tilde{\omega}^{A_1A_2}$.
\eit

This last point says that noisy resources can be used only when they
leak information to the environment that is independent of the inputs
or the particular protocol being run.  For example, if we are given
$[c \ra c]$ as part of the input resource $\alpha$ then we can use it
only if we promise to input the same distribution of 0 and 1
regardless of which protocol we're using.  On the other hand, $[q\ra
q]$ can be used with any input since it doesn't leak anything to the
environment.

The primary application of clean protocols is the
following general principle.
\begin{lemma}[protocol superposition principle]\label{lem:superposition}
Suppose that Alice and Bob would like to execute $m$ different
operations, $\cP_1,\ldots, \cP_m$ in superposition.  This means that
they would like to perform an operation $\cP$ that 
satisfies
$$\cP \sum_{k=1}^m c_k \ket{k}_A\ket{k}_B\ket{\psi}_{AB} \approx_{\eps_n}
\sum_{k=1}^m c_k \ket{k}_A\ket{k}_B
\cP_k \ket{\psi}_{AB}$$
for any coefficients $\{c_k\}$, and where $\eps_n\ra 0$ as
$n\ra\infty$. 

Let $\cR_k$ denote the set of resources capable of simulating $\cP_k$
{\em cleanly}:
$\cR_k := \l\{\alpha : \alpha \geqclean \<\cP_k\> \r\}. $
Then
$$\cR := \l\{\alpha : \alpha \geqclean \<\cP\>\r\} = \bigcap_k
\cR_k.$$
\end{lemma}
One direction of the proof is easy: $\cR\subseteq \cap_k \cR_k$, since 
$\cP\geqclean \cP_k$ for each $k$.  To prove $\cR \supseteq \cap_k
\cR_k$, we start with $\alpha\in\cap_k\cR_k$ and clean protocols for
$\alpha\geqclean \cR_k$ for each $k$.  Then have Alice and Bob run
each protocol conditioned on their (shared) value of $k$.  Since each
resource inequality is clean, nothing is discarded that would break
superpositions over different values of $k$.

{\em Remark:} Note that clean resource inequalities are not the same
as reversible RIs (meaning asymptotic equivalences; $\alpha\geq \beta$
and $\beta\geq \alpha$).  For example, while clean RIs cannot freely
discard entanglement, they may discard communication resources, or use
communication to reduce entanglement, neither of which are reversible.
On the other hand, most resource equalities can be made clean.  The
only possible complication arises for protocols that use an unlimited
amount of entanglement: i.e. $2[c\ra c] + \infty [qq] = [q\ra q] +
\infty[qq]$ can be made clean only with some additional effort.

\subsection{Approach: entanglement capacity as an interval}
\label{sec:ent-cap}

Armed with the definition of clean RIs, we now examine the entangling
capacities of various quantum operations.  The most obvious
restriction is that entanglement cannot be cleanly discarded, so that
e.g. while $2[qq]\geq [qq]$, it does not hold that $2[qq]\geqclean
[qq]$.  Moreover, eliminating entanglement cleanly is now a
non-trivial resource.  So
$$[c\ra c : I/2] \geqclean -[qq]$$
 via a
protocol where Alice sends her half of a shared state $\ket{\Phi}$
through the 
classical channel and Bob performs a CNOT with the bit he receives as
control and with his half of $\ket{\Phi}$ as target.  Then he is left
with a $\ket{0}$, which he discards.  At the same time, $[c\ra
c:I/2]\geqclean \emptyset$ (where $\emptyset$ is the null resource),
since Alice can always send a random bit through a channel.  This
means that both protocols can be run in superposition and classical
communication can be used to generate superpositions of different
amounts of entanglement.  To express this concisely, we can say that
the {\em entanglement capacity range} of $[c\ra c]$ (alternately, its
{\em spread capacity}) contains $[-1,0]$. In fact, the spread capacity
of $[c\ra c]$ is exactly $[-1,0]$ as can be seen from the non-increase
of entanglement under LOCC for the upper bound and \thmref{spread} for
the lower bound.

We can apply this approach to other resources as well.  Instead of
measuring the maximum entanglement that can be sent using a resource,
we will find the range of entanglent that it can cleanly
generate/consume.  Table~\ref{tab:spread} lists the entanglement
spread capacities of some common communication resources.
\begin{table}[h]
\caption{Various entanglement spread capacities.}\label{tab:spread} \begin{center}
  \begin{tabular}{|ll|c|c|}
\hline
resource & (abbr.) & min & max \\\hline
qubit & $[q\ra q]$ or $[q\la q]$ &  -1 & 1\\ \hline
cbit & $[c\ra c]$ or $[c\la c]$ & -1 & 0\\ \hline
cobit\cite{Har03} & $[q\!\ra\!qq]$ or $[qq\!\la\! q]$ & 0 & 1\\ \hline
co-cobit\cite{HS05} & $[q\!\la\!qq]$ or $[qq\!\ra\!q]$ & -1 & 0\\ \hline
ebit &  $[qq]$ & 1 & 1 \\ \hline
partially ent. states & $\ket{\psi}^{\otimes n}$ 
& $nE - O(\sigma \sqrt{n})$ & $nE + O(\sigma \sqrt{n})$ \\ \hline
embezzler\cite{HV03} & $\ket{\varphi_n}$ & $-n\epsilon$ & $n\epsilon$ \\\hline
unitary gate & $\<U\>$ & $-E(U^\dag)$ & $E(U)$ \\ \hline
\end{tabular}
\end{center}
\end{table}
Most of the bounds are straightforward to prove using \thmref{spread}
and the LOCC-monotonicity of entropy of entanglement, but a few lines
require explanation.  Cobits and co-cobits are defined in \cite{Har03}
and \cite{HS05} respectively, and their capacities can be proven by
reversibly mapping them to combinations of $[q\ra q]$ and $[qq]$.
For the partially entangled state $\ket{\psi}_{AB}$ with reduced
density matrices $\psi^A$ and $\psi^B$, we define $E = S(A)_\psi =
-\tr \psi^A \log \psi^A$ and $\sigma^2 = \tr \psi^A (\log \psi^A)^2 -
E^2$.
The $n$-qubit embezzling state\cite{HV03} $\ket{\varphi_n}$ is defined to be
(up to normalization) $\sum_{i=1}^{2^n}
\frac{1}{\sqrt{i}}\ket{i}^A\ket{i}^B$ and can be used catalytically
(unlike other resources, which are consumed) to create or destroy
$n\eps$ ebits (or indeed any state of Schmidt rank $\leq 2^{n\eps}$)
by incurring error $\eps$.  Finally, for a bipartite unitary gate $U$,
$E(U)$ is its entanglement capacity\cite{BHLS02}: $E(U) := \max\{e :
\<U\> \geq e [qq]\}$

By time-sharing, the spread capacity of a resource $\alpha$ is
completely characterized by an upper and lower bound (although there
can always be tradeoffs between entanglement and other resources).
And by \lemref{superposition} clean protocols using any amount of
entanglement within that range can be performed in superposition using
$\alpha$.

\section{Applications}\label{sec:applications}
This section will describe three applications of entanglement spread
and clean resource inequalities to some apparently unrelated problems.

\subsection{Entanglement dilution}\label{sec:dilution}
Let $\ket{\psi}^{AB}$ be a partially entangled state.  Then, as
mentioned above, $\Theta(\sqrt{n})$ cbits in either direction are
necessary\cite{HL02, HW02} and sufficient\cite{LP99} to prepare
$\ket{\psi}^{\ot n}$ from EPR pairs.  Unlike many other lower bounds
on communication, this bound holds for a task which itself has no
communication capacity: that is, that ability to create
$\ket{\psi}^{AB}$ from singlets has no value for
communication.\footnote{Being able to selectively either create
  $\ket{\psi}^{AB}$, or some other state, such as an all zeroes state,
  {\em would} have communication capacity if the decision of which
  state to create were influenced by one party.  However, this is
  different from the ability to create   $\ket{\psi}^{AB}$ from EPR
  pairs once both parties have agreed to perform this task.}

The framework of entanglement spread explains this strange situation
by the fact that $\Omega(\sqrt{n})$ spread needs to be created in
order to (approximately) prepare $\ket{\psi}^{AB}$.  Thus the apparent
requirement for communication is something of a red herring: for
example, using a $O(\sqrt{n}/\eps)$-qubit embezzling state, together
with $nE(\psi)$ EPR pairs, would be enough to prepare
$\ket{\psi}^{AB}$ up to error $\approx\!\eps$.  Alternatively, if we
used some quantum operation, such as a bipartite unitary $U$, to
prepare $\ket{\psi}^{\ot n}$ from $nE(\psi)$ EPR pairs, then the
number of uses of $U$ needed would be related to the entanglement
spread capacity of $U$ (discussed in \secref{cc-lb}) rather than by
its communication capacity.  Thus, the communication cost of
entanglement dilution can more fruitfully be understood as a spread
cost.

\subsection{Using entanglement capacity as a lower bound for
  communication complexity}\label{sec:cc-lb}

A common method of lower bounding the communication complexity of a
distributed function, or almost equivalently, the cost to simulate a
bipartite unitary gate, has been to use its capacity to
communicate\cites{CDNT98, MW07}. In fact, entanglement capacity, or
even better, entanglement spread capacity, is a lower bound that is
always at least as strong (since $E(U)$ is at least as large as the
communication capacity of $U$~\cite{BHLS02}).  The main idea is the
following lower bound for simulating a unitary gate:
\begin{theorem}\label{thm:sim-LB}
If $U$ is a bipartite unitary gate such that
$$\infty [qq] + Q_1 [q\ra q] + Q_2[q\la q] \geq \<U \>,$$
then $2(Q_1+Q_2)\geq E(U) + E(U^\dag)$.
\end{theorem}
\begin{proof}
The proof is a simple application of \thmref{spread}.  Starting with $m$
EPR pairs, we can use $U$ $n$ times to create a superposition of
roughly $m+nE(U)$ EPR pairs and $m-nE(U^\dag)$ EPR pairs (assuming
that $m$ is large enough), for a total spread of roughly
$n(E(U)+E(U^\dag))$.  On the other hand, each transmitted qubit can
increase spread by at most 2.\end{proof}

As a corollary, we obtain a new lower bound for the communication
complexity of a function $f(x,y)$ where Alice holds $x$ and Bob holds
$y$, even with Alice and Bob can use an unlimited number of EPR pairs.
Following \cite{MW07} we can turn any quantum protocol for 
computing $f(x,y)$ into a unitary gate that approximates
$U_f:=\sum_{x,y} (-1)^{f(x,y)} \proj{x}\ot \proj{y}.$ If the protocol
for $f$ is exact then this requires only running the protocol once
forward and once backwards (reversing all communication and inverting
all local unitaries).  In the bounded-error case, we could (following
\cite{MW07}) repeat the protocol for $f$ $O(\log n)$ times to reduce
the error to $1/n^2$, so that (from Lemma 1 of \cite{HS05}), the
resulting transform has capacities within $O(1/n)$ of those of $U_f$.
Thus $E(U_f)$ is a lower bound on the communication complexity of
computing $f$ exactly and $\Omega(E(U_f)/\log n)$ is a lower bound on
its bounded-error complexity.  In both cases the lower bounds hold
even when the protocols for computing $f$ can use an unlimited number
of EPR pairs.

These bounds resemble the results of \cite{MW07}, where the
entanglement-assisted communication capacity of $U_f$ was used to
lower bound the complexity of $f$; however, they are always at least
as powerful.  This is because of the identity $E(U)+E(U^\dag)\geq
C_+^E(U)$, where $C_+^E(U) = \max\{C_1+C_2 : \<U\> + \infty [qq] \geq
C_1[c\!\ra\!c] + C_2[c\!\la\!c]\}$ is the entanglement-assisted
communication capacity of $U$ \cite{BS03a, Har05}.  Moreover, there
exist\cite{HL08} unitary gates where $E(U)$ can be exponentially
larger than $C_+^E(U)$ and so the spread technique gives
correspondingly better bounds, although no distributed functions
corresponding to such unitaries are known.

There is one weakness to this lower bound technique, which
also appears in the related R\`{e}nyi-entropy-based lower bounds of
\cite{DH02}.  While the lower bounds apply even when unlimited EPR
pairs are allowed, they no longer hold when protocols for $f$ are
assisted by embezzling states or other entangled states with nonzero
spread.  However, in future work, I will prove that  communication
capacities differ by at most a constant factor between the
EPR-assisted and the embezzling-assisted cases. 

\subsection{The Quantum Reverse Shannon Theorem on general inputs}
\label{sec:qrst}

If Shannon's noisy channel coding theorem is thought of as using a
noisy channel to simulate a noiseless channel,
then the reverse Shannon theorem uses a noiseless channel to simulate
a noisy channel.  When free shared randomness is allowed, this
simulation can be performed using an asymptotic rate of communication
equal to the capacity of the channel\cite{BSST01, Winter:02a}.

Similarly, the idea of the quantum reverse Shannon theorem (or QRST;
originally proposed in \cite{BSST01}) is to simulate a quantum channel
using shared entanglement and a rate of communication equal to its
entanglement-assisted capacity.  So far this simulation has only been
shown possible when the inputs are tensor powers\cite{Devetak05a,
  ADHW06, BDHSW-qrst} and in a few other special cases.  However, in
\cite{BDHSW-qrst}, we will show that for general inputs a quantum channel
cannot be simulated using unlimited EPR pairs and a rate of
communication equal to its capacity.  In this section, I'll give a
sketch of why this is true.  The main problem is that simulating
channels generally requires creating linear amounts of spread, for
which an an additional resource is necessary, such as extra
communication (which could be from Bob to Alice), an embezzling
state or some other non-standard entangled state.

We can even use \lemref{superposition} to obtain the optimal rates for
the case of general sources.  To see this, we first review the rates
for the classical reverse Shannon theorem for general sources with
feedback (meaning that Alice learns Bob's output).  This simulation
task requires $C := \max_p I(A;B)_p$ cbits and $R := \max_p H(B)_p -
C$ rbits (bits of shared randomness)\cite{BSST01, Winter:02a}.  Here
$p$ is some input distribution, $I(A:B)_p$ is the mutual information
between random variables $A$ and $B$ where $A$ is distributed
according to $p$ and $B$ 
is given by passing $A$ through a noisy channel, and $H(B)_p$ is the
entropy of $B$ in the same setting.  Note that the rbit cost may be
smaller than $\max_p H(B)-I(A;B) = \max_p H(B|A)_p$ since for some
input distributions $I(A;B)_p$ may not be maximal, and thus some of
the classical communication can be used to substitute for randomness.
Thus, in considering the worst input distributions, we should look at
their cbit cost and their total cbit+rbit cost and maximize each
separately.  If we give up the feedback requirement then there is now
a somewhat more complicated cbit/rbit tradeoff curve, but the feasible
resource region for general sources is still simply the intersection
of the feasible resource regions of each possible i.i.d. source.

For the quantum reverse Shannon theorem this is no longer true, due
to the cost of creating entanglement spread.  For simplicity we will
discuss the QRST 
with feedback, meaning that, instead of working with a noisy channel
$\cN$, we use its isometric extension $U_\cN$ and give to Alice the
part of the output which normally would go to the environment.  When
we restrict $U_\cN$ to a source with average density matrix $\rho$,
then the optimal simulation protocol corresponds to the RI
$$ I(A;B)_\rho [c\ra c] + S(B)_\rho [qq] \geq
\<U_\cN : \rho\>.$$
Here the quantities $I(A;B)_\rho$ and $S(B)_\rho$ refer to a state in which
$\rho^{A'}$ is purified into a state $\ket{\Phi}^{AA'}$, and then $A'$
is sent through $U_\cN$ to obtain outputs for $B$ and $E$.  However,
for general sources---or even for locally distinguishable mixtures of
tensor power states---the communication cost is not simply $\max_\rho
I(A;B)_\rho$, even when unlimited EPR pairs are allowed.  The problem
is that the $H(B)_\rho$ entanglement cost varies from source to
source, and naively running each fixed-source protocol in
superposition will consume varying amounts of entanglement, leading to
decoherence between different branches of the superposition.

In order to run different fixed-source channel simulations in
superposition, we then need to use some source of entanglement
spread, such as extra communication or an embezzling state, to obtain
superpositions of different amounts of entanglement.  For concreteness,
suppose that we want
to simulate $\<U_\cN:\rho\>$ using $C_1 [c\ra c] + C_2 [c\la c] + E
[qq]$.  Then this is possible if and only if
\begin{align*}
C_1 &\geq \max_\rho I(A;B)_\sigma \qquad
E \geq \max_\rho H(B)_\sigma\\
C_2 &\geq E - \min_\rho \l[H(B)_\sigma + \min(0,
I(A;B)_\sigma-C_1)\r]
\end{align*}
Similar tradeoffs can be derived for any combination of input
resources. The proof that these resources are both necessary and
sufficient will follow from \lemref{superposition}; we need only the
additional fact that any simulation protocol can be converted into a
clean simulation protocol, modified only by the possible addition of
some pure entanglement in the output.

One additional tool will be
needed for completely general inputs, as opposed to mixtures of tensor
powers: by analogy to the classification of bit strings by their types
(i.e. frequencies of different symbols), the full QRST will decompose
the input in the Schur basis, which splits it into different
irreducible representations of the unitary and symmetric groups.  Full
details of the construction and its proof of optimality will be in
\cite{BDHSW-qrst}.

\section{Conclusion}
The problem of entanglement spread shows that even bipartite pure state
entanglement can retain some surprises.
Moreover we have seen in the QRST that the idea of spread
can yield precise and non-trivial statements even when dealing with
mixed states and noisy channels.

One major unresolved question about spread relates to its
interconvertibility.  While some resources, such as classical
communication and embezzling states, are more or less universal
sources of spread, it is still possible that spread exists in many
incomparable forms.  For example, consider the problem of transforming
$\ket{\psi_1}^{\ot m}$ into $\ket{\psi_2}^{\ot n}$ for some integers
$m,n$ and partially entangled states $\ket{\psi_1}, \ket{\psi_2}$.  If
Alice and Bob can communicate classically (even $O(\sqrt{n})$ bits)
then from \cite{BBPS96} we know that it suffices to use
$m=nE(\psi_2)/E(\psi_1) + o(n)$ copies of $\ket{\psi_1}$, and it is
necessary to use at least $m=nE(\psi_2)/E(\psi_1) -o(n)$ copies.
But is this amount of communication always necessary?
However, the current formulations of entanglement spread give only
partial answers to this question: in some cases $\Omega(\sqrt{n})$
classical communication is necessary\cite{FL05}, but it is unknown if
this holds generically, and how it depends on $m$.

A more challenging problem is to
determine the extent to which shared EPR pairs or shared embezzling states can reduce the quantum communication complexity for a
distributed function.   Many examples are known in which access to EPR pairs can dramatically reduce the {\em classical} communication complexity\cite{Gavinsky09}, but much less is known about the question of quantum communication complexity.  The largest known separation between the quantum communication complexity with and without shared EPR pairs is a mere factor of two, in the case of the equality function~\cite{Winter-message-ID}.  On the other hand, the only known upper bound on the usefulness of shared entanglement is that the classical communication complexity without shared entanglement is at most exponentially greater than the (quantum or classical) communication complexity with shared entanglement.

A related question, originally posed in \cite{DH02}, is to determine whether shared embezzling states can ever be more useful in reducing communication complexity than EPR pairs can.  In a future paper I will prove that communication complexity with shared EPR pairs is at most a constant factor higher than with embezzling states (or indeed any shared entangled state).

{\em Acknowledgments:} I want to thank Charlie Bennett, Patrick
Hayden, Debbie Leung, Peter Shor and Andreas Winter for many useful
conversations on this subject.  My funding is from the Army Research
Office under grant W9111NF-05-1-0294, the European Commission under
Marie Curie grants ASTQIT (FP6-022194) and QAP (IST-2005-15848), and
the U.K. Engineering and Physical Science Research Council through
``QIP IRC.''

\bibliographystyle{IEEEtran.bst}
\bibliography{../../latex/bib}

\begin{thebibliography}{10}
\providecommand{\url}[1]{#1}
\csname url@rmstyle\endcsname
\providecommand{\newblock}{\relax}
\providecommand{\bibinfo}[2]{#2}
\providecommand\BIBentrySTDinterwordspacing{\spaceskip=0pt\relax}
\providecommand\BIBentryALTinterwordstretchfactor{4}
\providecommand\BIBentryALTinterwordspacing{\spaceskip=\fontdimen2\font plus
\BIBentryALTinterwordstretchfactor\fontdimen3\font minus
  \fontdimen4\font\relax}
\providecommand\BIBforeignlanguage[2]{{%
\expandafter\ifx\csname l@#1\endcsname\relax
\typeout{** WARNING: IEEEtran.bst: No hyphenation pattern has been}%
\typeout{** loaded for the language `#1'. Using the pattern for}%
\typeout{** the default language instead.}%
\else
\language=\csname l@#1\endcsname
\fi
#2}}

\bibitem{Bennett04}
C.~H. Bennett, ``A resource-based view of quantum information,'' \emph{Quantum
  Inf. Comput.}, vol.~4, no. 6\&7, pp. 460--466, 2004.

\bibitem{DHW05}
I.~Devetak, A.~W. Harrow, and A.~Winter, ``A resource framework for quantum
  {S}hannon theory,'' \emph{IEEE Trans. Inf. Theory}, vol.~54, no.~10, pp.
  4587--4618, Oct 2008, arXiv:quant-ph/0512015.

\bibitem{BBPS96}
C.~H. Bennett, H.~J. Bernstein, S.~Popescu, and B.~Schumacher, ``Concentrating
  partial entanglement by local operations,'' \emph{Phys. Rev. A}, vol.~53, pp.
  2046--2052, 1996, quant-ph/9511030.

\bibitem{LP99}
H.-K. Lo and S.~Popescu, ``The classical communication cost of entanglement
  manipulation: Is entanglement an inter-convertible resource?''
  \emph{Phys.~Rev.~Lett.}, vol.~83, pp. 1459--1462, 1999,
  arXiv:quant-ph/9902045.

\bibitem{HW02}
P.~Hayden and A.~Winter, ``On the communication cost of entanglement
  transformations,'' \emph{pra}, vol.~67, p. 012306, 2003,
  arXiv:quant-ph/0204092.

\bibitem{HS05}
A.~W. Harrow and P.~Shor, ``Time reversal and exchange symmetries of unitary
  gate capacities,'' \emph{IEEE Trans. Inf. Theory}, vol.~56, no.~1, pp.
  462--475, Jan 2010, arXiv:quant-ph/0511219.

\bibitem{Har03}
A.~W. Harrow, ``Coherent communication of classical messages,''
  \emph{Phys.~Rev.~Lett.}, vol.~92, p. 097902, 2004, arXiv:quant-ph/0307091.

\bibitem{HV03}
P.~Hayden and W.~van Dam, ``Universal entanglement transformations without
  communication,'' \emph{pra}, vol.~67, p. 060302(R), 2003, quant-ph/0201041.

\bibitem{BHLS02}
C.~H. Bennett, A.~W. Harrow, D.~W. Leung, and J.~A. Smolin, ``On the capacities
  of bipartite {H}amiltonians and unitary gates,'' \emph{IEEE Trans. Inf.
  Theory}, vol.~49, no.~8, pp. 1895--1911, 2003, arXiv:quant-ph/0205057.

\bibitem{HL02}
A.~W. Harrow and H.-K. Lo, ``A tight lower bound on the classical communication
  cost of entanglement dilution,'' \emph{IEEE Trans. Inf. Theory}, vol.~50,
  no.~2, pp. 319--327, 2004, arXiv:quant-ph/0204096.

\bibitem{CDNT98}
R.~Cleve, W.~{van Dam}, M.~Nielsen, and A.~Tapp, ``Quantum entanglement and the
  communication complexity of the inner product function,'' \emph{Lect. Notes
  Comput. Sci.}, vol. 1509, p.~61, 1998, quant-ph/9708019.

\bibitem{MW07}
A.~Montanaro and A.~Winter, ``A lower bound on entanglement-assisted quantum
  communication complexity,'' in \emph{{ICALP '07: Proc. Intl. Coll. on
  Automata, Languages and Programming}}, vol. 4596/2007.\hskip 1em plus 0.5em
  minus 0.4em\relax Springer Berlin/Heidelberg, 2007, pp. 122--133,
  arXiv:quant-ph/0610085.

\bibitem{BS03a}
D.~W. Berry and B.~C. Sanders, ``Relations for classical communication capacity
  and entanglement capability of two-qubit operations,'' \emph{Phys. Rev. A},
  vol.~67, p. 040302(R), 2003, quant-ph/0205181.

\bibitem{Har05}
A.~W. Harrow, ``Applications of coherent classical communication and {S}chur
  duality to quantum information theory,'' Ph.D. dissertation, M.I.T.,
  Cambridge, MA, 2005, arXiv:quant-ph/0512255.

\bibitem{HL08}
A.~W. Harrow and R.~Low, ``Efficient quantum tensor product expanders and
  $k$-designs,'' in \emph{Proc. of APPROX-RANDOM}, ser. LNCS, vol. 5687.\hskip
  1em plus 0.5em minus 0.4em\relax Springer, 2009, pp. 548--561,
  arXiv:0811.2597.

\bibitem{DH02}
W.~{van Dam} and P.~Hayden, ``Renyi-entropic bounds on quantum communication,''
  2002, quant-ph/0204093.

\bibitem{BSST01}
C.~H. Bennett, P.~W. Shor, J.~A. Smolin, and A.~Thapliyal,
  ``Entanglement-assisted capacity of a quantum channel and the reverse
  {S}hannon theorem,'' \emph{IEEE Trans. Inf. Theory}, vol.~48, pp. 2637--2655,
  2002, quant-ph/0106052.

\bibitem{Winter:02a}
A.~Winter, ``Compression of sources of probability distributions and density
  operators,'' 2002, quant-ph/0208131.

\bibitem{Devetak05a}
I.~Devetak, ``Triangle of dualities between quantum communication protocols.''
  \emph{prl}, vol.~97, p. 140503, 2006, arXiv:quant-ph/0505138.

\bibitem{ADHW06}
A.~Abeyesinghe, I.~Devetak, P.~Hayden, and A.~Winter, ``The mother of all
  protocols: Restructuring quantum information's family tree,'' \emph{Proc.
  Roc. Soc. A}, vol. 465, no. 2108, pp. 2537--2563, 2009,
  arXiv:quant-ph/0606225.

\bibitem{BDHSW-qrst}
C.~Bennett, I.~Devetak, A.~W. Harrow, P.~Shor, and A.~Winter, ``The quantum
  reverse {S}hannon theorem,'' 2009, arXiv:0912.5537.

\bibitem{FL05}
B.~Fortescue and H.-K. Lo, ``Inefficiency and classical communication bounds
  for conversion between partially entangled pure bipartite states,''
  \emph{Phys. Rev. A}, vol.~72, 2005, quant-ph/0411200.

\bibitem{Gavinsky09}
D.~Gavinsky, ``Classical interaction cannot replace quantum nonlocality,''
  2009, arXiv:0901.0956.

\bibitem{Winter-message-ID}
A.~Winter, ``Identification via quantum channels in the presence of prior
  correlation and feedback,'' 2004, arXiv:quant-ph/0403203.

\end{thebibliography}
\end{document}